\documentclass[11pt,letter]{article}
\usepackage[hidelinks]{hyperref} 
\usepackage[margin=1in]{geometry}
\usepackage[utf8x]{inputenc}
\usepackage{amsthm}
\usepackage{todonotes,wrapfig,float,graphicx,amssymb,textcomp,array,amsmath}
\usepackage{enumerate,enumitem}
\usepackage{multirow}
\usepackage{tabularx}
\usepackage{color,xcolor}

\definecolor{mycolor}{rgb}{0, 0, 0}
\usepackage[titletoc,title]{appendix}

\newcommand{\OPT}{\mathrm{OPT}}
\newcommand{\boundary}[1]{\partial #1}

\newcommand{\angseq}[1]{\text{\em A}(#1)}

\title{Contiguous Boundary Guarding}
\author{
\and Ahmad Biniaz\thanks{School of Computer Science, University of Windsor, \texttt{abiniaz@uwindsor.ca}. Research supported in part by NSERC.}
\and Anil Maheshwari\thanks{School of Computer Science, Carleton University, \texttt{anil@scs.carleton.ca}. Research supported in part by NSERC.} 
\and Joseph S. B. Mitchell\thanks{Department of Applied Mathematics and Statistics, 
Stony Brook University, \texttt{joseph.mitchell@stonybrook.edu}. Supported by the National Science Foundation (CCF-2007275).} 
\and Saeed Odak\thanks{School of Electrical Engineering and Computer Science, University of Ottawa, \texttt{saeedodak@gmail.com}. Research supported by NSERC.} 
\and Valentin Polishchuk\thanks{Communications and Transport Systems, Link\"{o}ping Univeristy, \texttt{valentin.polishchuk@liu.se}. Partially supported by the Swedish Transport Administration and the Swedish
Research Council.} 
\and Thomas Shermer\thanks{School of Computing Science, Simon Fraser University, \texttt{shermer@sfu.ca}. Research supported by NSERC.} }

\newif\ifacr\acrtrue

\newtheorem{lemma}{Lemma}
\newtheorem{theorem}{Theorem}
\newtheorem{observation}{Observation}
\newtheorem{corollary}{Corollary}

\begin{document}
\maketitle\begin{abstract}
We study the problem of guarding the boundary of a simple polygon with a minimum number of guards such that each guard covers a contiguous portion of the boundary. First, we present a simple greedy algorithm for this problem that returns a guard set of size at most $\OPT + 1$, where $\OPT$ is the number of guards in an optimal solution. Then, we present a polynomial-time exact algorithm. While the algorithm is not complicated, its correctness proof is rather involved. This result is interesting in the sense that guarding problems are typically NP-hard and, in particular, it is NP-hard to minimize the number of guards to see the boundary of a simple polygon, without the contiguous boundary guarding constraint.

From the combinatorial point of view, we show that any $n$-vertex polygon can be guarded by at most {\color{mycolor}$\lfloor \frac{n-2}{2}\rfloor$} guards. This bound is tight because there are polygons that require this many guards.
\end{abstract}

\section{Introduction}
The {\em art gallery problem}, introduced in 1973 by Victor Klee \cite{chvatal1975combinatorial,ORourke1987}, asks for a minimum number of guards that see every point of a given polygon $P$---the guards can lie anywhere in the polygon.\footnote{It is also known as the {\em point-guard} art gallery problem.} This problem is central to computational geometry and is still an active research area. It has been long known NP-hard \cite{Lee1986}; see also \cite{Aggarwal1984, ORourke1987,o1983some}. The recent breakthrough result by Abrahamsen, Adamaszek, and Miltzow \cite{Abrahamsen2022} shows that the problem is $\exists\mathbb{R}$-complete\footnote{$\exists\mathbb{R}$ consists of all problems reducible to the decision problem for the existential theory of the reals.} even if the corners of the polygon are at integer coordinates. The problem is notoriously difficult as the best known approximation algorithms have logarithmic factors \cite{bonnet2017approximation,Deshpande2007,Efrat2006}; see also \cite{Eidenbenz2001} for some lower bounds on the approximation factor.

The {\em boundary guarding} is a variant of the art gallery problem in which the goal is to guard only the boundary of $P$ with a minimum number of guards that can lie anywhere in $P$. One may think of this as guarding only the walls of an art gallery with no sculptures. By extending ideas of \cite{Abrahamsen2022}, Stade has shown that this variant is also $\exists\mathbb{R}$-complete \cite{stade2022point}. We introduce the {\em contiguous} version of this problem. In the {\em contiguous boundary guarding} problem, the goal is to guard the boundary of $P$ such that each guard covers a contiguous portion of the boundary. While the visibility polygon of a guard may have several connected components when intersected with the boundary, the guard is assigned to only one  component. 

Contiguous boundary guarding is natural, particularly if the guards have a bounded field of view or cannot rotate to cover all visible components on the boundary. This problem appeared in Open Problems from CCCG 2024 \cite{CCCG2024}, motivated by the fact that the hardness proofs for typical art gallery variants \cite{Abrahamsen2022,stade2022point} require guards to cover several connected components on the boundary.
In this paper, we present an exact polynomial-time algorithm and tight combinatorial bounds for the contiguous boundary guarding problem. 

\subsection{Our results}
We study the contiguous 
boundary guarding problem from combinatorial and computational points of view. In Section~\ref{bound-section} we give a constructive proof that for every integer $n\ge 4$, every $n$-vertex polygon can be guarded by at most {\color{mycolor}$\lfloor \frac{n-2}{2}\rfloor$ guards.
This bound is the best possible as some polygons require this many guards. 
}

In Section~\ref{greedy-section} we present a simple greedy algorithm that returns a guard set of size at most $\OPT+1$, where $\OPT$ is the number of guards in an optimal solution.\footnote{A similar greedy algorithm has been used in \cite{Abrahamsen2025} for some covering problems.} The algorithm starts from an arbitrary point $p$ on the boundary and then covers the boundary in a counter-clockwise (ccw) direction from $p$. It places a guard to cover $p$ and a maximal contiguous portion of the boundary in the ccw direction. The algorithm then continues greedily, covering a maximal contiguous boundary portion with each choice of guard, until returning to $p$.  
Despite getting very close to optimal by this greedy algorithm, achieving an optimal solution (avoiding an extra guard) is challenging. An interesting property of the greedy algorithm is that if $p$ is covered by two guards in some optimal solution, then the greedy algorithm returns a guard set of size $\OPT$.

In Section~\ref{optimal-section} we present a polynomial-time exact algorithm for the problem. This algorithm is not complicated, though its correctness proof (presented in Section~\ref{proof-section}) involves some nontrivial arguments. The main idea is first to find a polynomial-size set $S$ of points on the boundary such that at least one of them is covered by two guards in some optimal solution. Then, we find an optimal solution by running the greedy algorithm from every point in $S$. Finding such a set $S$ is the most technical part of the algorithm and the proof.


\subsection{Related work and results}
The art gallery problem and its variants have been extensively studied. Most of the variants (that arise from enforcing constraints, say, on the input polygon, on the location of the guards, or on the portion of the polygon to be guarded) remain NP-hard.

\subparagraph*{Hardness and approximations}
The {\em vertex-guard problem} is a discrete version of the art gallery problem  where the guards must be located at the vertices of the polygon. This version is also NP-hard \cite{Lee1986,o1983some} with known logarithmic factor \cite{Ghosh10} and sub-logarithmic factor \cite{King2011bd} approximation algorithms. A similar sub-logarithmic  approximation is known for the case where the guards can lie anywhere on the boundary of the polygon \cite{King2011bd}. Better approximation algorithms are known for monotone polygons \cite{krohn2013approximate} and weakly-visible polygons \cite{ashur2022terrain,bhattacharya2017approximability}.

Terrain guarding is a related problem in which we are given a 1.5-dimensional terrain (i.e., an $x$-monotone polygonal chain), and the goal is to find a minimum number of guards on the terrain that guards the terrain---this is a variant of the boundary guarding problem. This problem is also NP-hard \cite{Chen2095,King2011terrain}, and admits polynomial time approximation scheme~\cite{ashur2022terrain,Gibson2014}.

Both versions of the problem (point-guard and vertex-guard) remain NP-hard in orthogonal polygons\footnote{Also known as rectilinear polygons.} \cite{Schuchardt1995}. Even if we want to guard only the vertices of an orthogonal polygon, the problem is still NP-hard \cite{Katz2008} for three versions where the guards lie on vertices, on the boundary, or anywhere in the polygon.

\subparagraph*{Polynomial-time exact algorithms}
For some trivial instances, the art gallery problem can be solved in polynomial time, such as for convex, star-shaped, and spiral polygons. Finding non-trivial instances that admit polynomial-time algorithms is of particular interest.
To the best of our knowledge, there are only a few non-trivial instances for which the problem can be solved optimally in polynomial time. These instances usually enforce constraints on the polygon itself or on the definition of visibility. For example, the authors in~\cite{daescu2019altitude} present polynomial-time algorithms when the polygon is uni-monotone (i.e.~it is monotone and one of
    its two monotone polygonal chains is a straight line segment)\footnote{This type of polygons are also referred to as monotone mountains.} and for 1.5-dimensional terrains where all the guards must be placed on a horizontal line above the terrain.

The art gallery problem can be solved in polynomial time on orthogonal polygons but under {\em rectilinear} visibility \cite{worman2007polygon} and {\em staircase} visibility \cite{motwani1988covering}. In the former case, two points are visible if their minimum orthogonal bounding box lies in the polygon, and in the latter case, two points are visible if there is an $x$- and $y$-monotone path between them in the polygon.

All the above algorithms \cite{daescu2019altitude,motwani1988covering,worman2007polygon} are obtained by using the notion of {\em perfectness} of $P$, which means that the minimum number of guards is equal to the maximum number of points that can be placed in $P$ such that their visibility polygons are pairwise disjoint. This property, however, does not hold for our problem.

\subparagraph*{Combinatorial bounds}
While most variants of the art gallery problem are hard with high computational complexities, tight combinatorial bounds are known on the numbers of guards that are always sufficient and sometimes necessary to guard a polygon with $n$ vertices.
For the original art gallery problem the bound is $\lfloor n/3\rfloor$ for both point-guards and vertex-guards~\cite{chvatal1975combinatorial,fisk1978short}, and it is $\lfloor n/4\rfloor$ for orthogonal polygons   \cite{kahn1983traditional,lubiw1985decomposing,sack1988guard,urrutia2000art}. Tight bounds are also known for variants where an entire edge or diagonal can serve as a guard \cite{Biniaz2024,ORourke1983,shermer1992recent}.

\subparagraph*{Contiguity} The contiguous guarding is in particular related to guarding polygons with cameras/guards that have a bounded field of view (i.e.~the maximum angle a camera can cover is limited)~\cite{toth2000art,toth2002art} and to the so-called {\em floodlight} \cite{ORourkeSS95,Toth03} and  {\em city guarding} problems~\cite{Bao2008,Biniaz2023,Daescu2021}. It is also related to contiguous guarding of points on a 1.5-dimensional terrain by placing at most $k$ watchtowers such that each watchtower sees a contiguous sequence of points~\cite{kang4850503guarding}.

\subsection{Preliminaries}

Let $P$ be a simple polygon, i.e. a polygon with
    no holes and no self-intersections. We denote the boundary of $P$ by $\boundary{P}$. We define a {\em chain} as a connected portion of $\partial P$. A {\em reflex vertex} of $P$ is a vertex with an interior angle greater than $\pi$. A vertex of $P$ with an interior angle less than or equal to $\pi$ is called a {\em convex vertex}.

A {\em triangulation} of $P$ is a partition of the interior of $P$ into a set of triangles whose corners are at vertices of $P$. It is well-known that the dual of any triangulation of an $n$-vertex polygon is a binary tree with $n-2$ nodes. The {\em diameter} of a tree is the number of edges of a longest path in the tree.

\section{Combinatorial Bounds}
\label{bound-section}
In this section we determine the number of guards that are always sufficient and sometimes necessary to guard an $n$-vertex polygon.

It is easily seen that a chain with $n$ edges can be covered by {\color{mycolor}$\lfloor\frac{n+1}{2}\rfloor$} guards
---this can be achieved by placing a guard at every second vertex. 
Thus it follows that a polygon with $n$ vertices (and hence $n$ edges) can be covered by  {\color{mycolor}$\lfloor\frac{n+1}{2}\rfloor$} guards.
As we will see later
, this bound is very close to the best achievable bound but is not tight. We need stronger ingredients, such as the following lemma, to obtain a tight bound. In this section we say that a vertex ``covers'' an edge if the vertex sees both endpoints of the edge.

\begin{lemma}
\label{cover-lemma}
    For any polygon $P$ with at least $8$ vertices, one of the following statements~holds.
    
    \begin{enumerate}
        \item There is a vertex that covers $6$ consecutive edges of $\partial P$.
        \item There are two  vertices that collectively cover $7$ consecutive edges of $\partial P$.
        \item There are two  vertices that collectively cover $8$ edges of $\partial P$ such that each vertex covers $4$ consecutive edges.
    \end{enumerate}
\end{lemma}
\begin{proof}
    Triangulate $P$ and let $T$ be the dual tree of the triangulation. The largest vertex-degree in $T$ is at most $3$. Since $n\ge 8$, $T$ has at least $6$ nodes, and hence its diameter is at least $3$. The diameter is determined by a path, say $\delta$, with at least $4$ nodes where the two end-nodes, say $p$ and $q$, are leaves.
    Let $p'$ and $q'$ be the unique neighbors of $p$ and $q$ in $T$, respectively. We consider the following three cases that are depicted in Figure~\ref{degree-fig}.

\begin{figure}[htb]
	\centering
\setlength{\tabcolsep}{0in}
	$\begin{tabular}{ccc}
	\multicolumn{1}{m{.333\columnwidth}}{\centering\vspace{0pt}\includegraphics[width=.31\columnwidth]{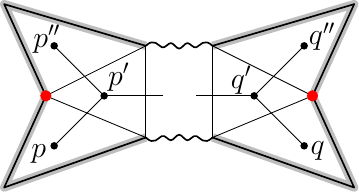}}
	&\multicolumn{1}{m{.333\columnwidth}}{\centering\vspace{0pt}\includegraphics[width=.31\columnwidth]{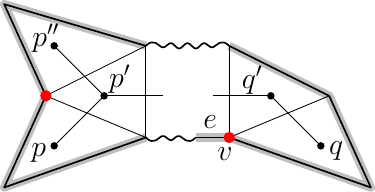}}&
 \multicolumn{1}{m{.333\columnwidth}}{\centering\vspace{20pt}\includegraphics[width=.31\columnwidth]{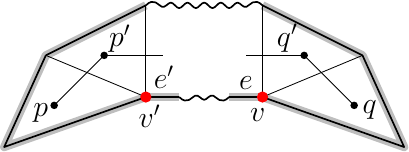}}
		\\
		(a)   &(b)&(c) 
	\end{tabular}$	
\caption{Illustration of the proof of Lemma~\ref{cover-lemma}.}
\label{degree-fig}
\end{figure}
   
    \begin{itemize}
        \item {\em Both $p'$ and $q'$ have degree 3.} This case is depicted in  Figure~\ref{degree-fig}(a). Since $\delta$ is a longest path in the tree, $p'$ has a neighbor $p''$ in $T$, different from $p$, that is also a leaf.  Since $p$ and $p''$ are leaves, the three triangles that correspond to $p,p',p''$ consist of $4$ consecutive edges of $\partial P$ that are visible from their middle vertex. A similar argument holds for $q'$. Since the diameter is at least 3, the four edges that we get for $q'$ are different from those for $p'$. The two middle vertices cover the $8$ edges, and thus statement 3 holds.
        \item {\em One of $p'$ and $q'$ has degree 3 and the other has degree 2.} Due to symmetry assume that $p'$ has degree $3$, as in  Figure~\ref{degree-fig}(b). As in the previous case, we get four consecutive edges of $\partial P$ for $p'$ that can be covered by their middle vertex. The two triangles that correspond to  $q,q'$ consist of three consecutive edges of $\partial P$ where one of their vertices, say $v$, is incident to both triangles. The vertex $v$ covers these three edges and another edge $e$ on the remaining portion of $\partial P$. If $e$ is among the four edges of $p'$, then statement 2 holds; otherwise statement 3 holds.
        \item {\em Both $p'$ and $q'$ have degree 2.} See  Figure~\ref{degree-fig}(c). Define $e$, $v$, and its four covering edges as in the previous case. Let $e'$ and $v'$ be the analogous edge and vertex for $p'$. If $v=v'$, then $v$ covers six consecutive edges; thus, statement 1 holds. Assume that $v\neq v'$. Then, the four edges of $p'$ and the four edges of $q'$ share at most one edge, which could be $e$ or $e'$. If they share an edge then statement 2 holds, otherwise statement 3 holds.\qedhere
    \end{itemize}
\end{proof}

\begin{theorem}
In any polygon $P$ with $n\ge 4$ vertices there exists a set of at most $\lfloor \frac{n-2}{2} \rfloor$ points that guard the boundary of $P$ contiguously.
   This bound is tight for even and odd $n$.
\end{theorem}


\begin{proof}
If $n=4,5$, then one vertex of the polygon covers the entire boundary (to verify this, observe that in any triangulation of a $5$-gon, there is a vertex that is incident to the three resulting triangles). If $n=6$, then any triangulation contains a diagonal that partitions the boundary of $P$ into 2 and 4 consecutive edges;\footnote{For any $k\ge 2$ there is a diagonal that cuts off at least $k$ and at most $2k {-}2$ edges of the polygon; see~\cite{Biniaz2024}.} each partite can be covered by one vertex. If $n=7$, then there is a diagonal that partitions the boundary of $P$ into 3 and 4 consecutive edges; again, each partite can be covered by one vertex.   

Assume that $n\ge 8$. Then, at least one of the cases in Lemma~\ref{cover-lemma} holds. For each case, we show how to cover $\partial P$ by the number of guards that is claimed.

\begin{itemize}
    \item {\em Statement 1 holds.} We cover $6$ consecutive edges by $1$ guard and the remaining chain of $n{-}6$ edges by at most $\lfloor\frac{n-5}{2}\rfloor$ guards. Thus, the total number of guards is at most $\lfloor\frac{n-3}{2}\rfloor$. 
    \item {\em Statement 2 holds.} We cover $7$ consecutive edges by $2$ guards and the remaining chain of $n{-}7$  edges by at most $\lfloor\frac{n-6}{2}\rfloor$ guards.
    The total number of guards is at most {\color{mycolor}$\lfloor\frac{n-2}{2}\rfloor$}.
    \item {\em Statement 3 holds.} We cover the $8$ edges by two guards. If $8$ edges are consecutive, then we cover the remaining chain as in the previous case. Assume they are not consecutive and thus split into two sets, each having four consecutive edges. The remaining $n{-}8$ edges of $\partial P$ form two disjoint chains of $n_1$ and $n_2$ edges that can be covered by at most $\lfloor\frac{n_1+1}{2}\rfloor$ and $\lfloor\frac{n_2+1}{2}\rfloor$ guards, respectively. The total number of guards is at most {\color{mycolor}$\lfloor\frac{n_1+1}{2}\rfloor+\lfloor\frac{n_2+1}{2}\rfloor+2\le\frac{n-2}{2}$. 
    This equals $\lfloor\frac{n-2}{2}\rfloor$ when $n$ is even. When $n$ is odd, one of the chain sizes, say $n_1$, is even, and hence the chain itself can be covered by at most $\lfloor\frac{n_1}{2}\rfloor$ guards, which would result in at most $\frac{n-3}{2}\le \lfloor\frac{n-2}{2}\rfloor$ total guards.}
\end{itemize}

\begin{figure}[htb]
	\centering
\setlength{\tabcolsep}{0in}	\includegraphics[width=.6\columnwidth]{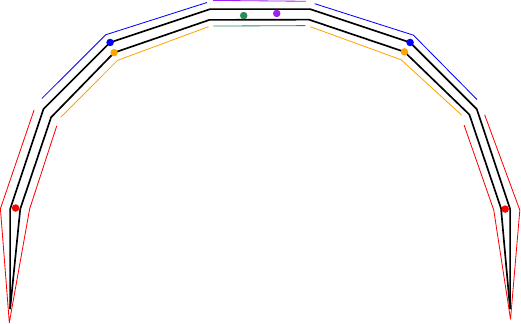}\\	\caption{Illustration of the lower bound $\lfloor\frac{n-2}{2}\rfloor$ guards. The polygonal arcs around the boundary are maximal chains that can each be guarded by a single guard.}
\label{lower-bound-fig}
\end{figure}

To verify the lower bound, see the polygon in Figure~\ref{lower-bound-fig} of size $n=4k+2$, for some integer $k$, which consists of two convex chains of odd length $\frac{n}{2}$. The chains are placed close to each other such that the midpoints of no three consecutive edges on a chain are visible from the same point in the polygon.  A guard cannot cover more than two consecutive midpoints on each chain. To cover each chain, we need at least $\lfloor\frac{n/2+1}{2}\rfloor$ guards, out of which only two guards that cover the two endpoints of the chain can be shared. Thus, we need at least $2\lfloor\frac{n/2+1}{2}\rfloor-2=\lfloor\frac{n-2}{2}\rfloor$ guards to cover the boundary of this polygon contiguously. This verifies the tightness of the bound for even $n$.
By adding a vertex on any edge of the polygon in Figure~\ref{lower-bound-fig}, a lower bound example for odd $n$ is obtained. 
\end{proof}

\section{Preliminaries for the Algorithms}
For two points $p$ and $q$ in the plane, we denote by $pq$ the straight line segment between $p$ and $q$. A {\em ray} from $p$ toward $q$, denoted $\overrightarrow{pq}$, is the half-line
that starts from $p$ and passes through $q$. A {\em wedge} is a region of the plane that is bounded by two rays starting from the same point known as the {\em apex} of the wedge. Notice that such two rays partition the plane into two wedges. 

Let $P$ be a simple polygon. Two
points $p$ and $q$ of $P$ are said to be {\em visible} (or {\em see} each other) if the line segment $pq$ lies
totally inside (in the interior or on the boundary of) $P$. The {\em visibility
region} of $P$ from a point $q\in P$, denoted $V(P,q)$, is the set of all points of $P$ that are visible
from $q$; see Figure~\ref{guard-fig}. It is well-known that $V(P,q)$ is a simple polygon and can be computed in linear time, see e.g.~\cite{ElGindy1981,Lee1983,Preparata1985}.  {\color{mycolor}The region $V(P,q)$ is also referred to as the {\em visibility polygon} of $q$.
The visibility polygon is {\em star-shaped}, i.e., all its points are visible from a single point called a {\em center}. The set of all centers of a star-shaped polygon is called its {\em kernel}; it can be computed in linear time~\cite{Lee1979}.}

\begin{figure}[!ht]
	\centering
\setlength{\tabcolsep}{0in}	\includegraphics[width=.5\columnwidth]{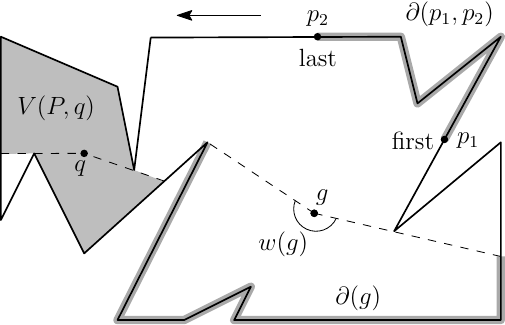}\\	\caption{The boundary of the polygon is assumed to be directed counter-clockwise.}
\label{guard-fig}
\end{figure}

To simplify our notation, we write $\partial$ for $\boundary{P}$. 
{\color{mycolor}We assume that $\partial$ is a directed cycle, with vertices ordered counter-clockwise.}
For an ordered pair $(p_1,p_2)$ of two points on $\partial$, we denote by $\partial(p_1,p_2)$  the portion of $\partial$ from $p_1$ to $p_2$ in counter-clockwise direction---$\partial(p_1,p_2)$ is a directed path. Observe that $\partial$ is the union of $\partial(p_1,p_2)$ and $\partial(p_2,p_1)$. We refer to $p_1$ and $p_2$ as the {\em first} and the {\em last} endpoints of $\partial(p_1,p_2)$, respectively. Since $\partial(p_1,p_2)$ is directed, for any two points $p$ and $q$ on $\partial(p_1,p_2)$, we can tell if $p$ appears {\em before} or {\em after} $q$.

Let $g$ be a guard in some contiguous boundary guarding of $\partial$. 
{\color{mycolor}Let $\partial(g)$ denote a connected subset of $\partial$ that is covered by $g$ (refer to Figure~\ref{guard-fig}); in other words $\partial(g)$ is portion of $\partial$ that is assigned to $g$ by the guarding.} Let $p_1(g)$ and $p_2(g)$ denote the first and last endpoints of $\partial(g)$.
We denote by $w(g)$ the wedge from the ray $\overrightarrow{gp_1(g)}$ to the ray
    $\overrightarrow{gp_2(g)}$. We refer to the angle at the apex of $w(g)$ as the {\em covering angle} of $g$; when it is clear from the context we use $w(g)$ to refer to this angle.  
{\color{mycolor}Let $\Gamma$ be a contiguous guarding of $\partial$, which is a set of guards with a connected subset of $\partial$ assigned to every guard. We say that $\Gamma$ is {\em minimal} if the removal of any guard from $\Gamma$ would result in an uncovered portion of  $\partial$. Thus, for every guard $g$ in a
    minimal contiguous guarding, there exists a point $g_p$ of $\partial$ that is
    guarded by only $g$, 
    because otherwise we could remove $g$ from the guard set. Without further mentioning, all guard sets considered in this paper are 
minimal. 
For any (minimal) contiguous guarding $\Gamma$, the
    counter-clockwise order of the $\{g_p : g \in \Gamma\}$ along $\partial$ induces an
    order of the guards in $\Gamma$.
Thus, the terms {\em previous} guard, {\em next} guard, and {\em consecutive} guards are well defined.} The following simple observation is valid for any guarding with at least two guards.

\begin{observation}
    \label{endpoint-obs}
    Let $g$ be a guard in a guarding of $\partial$. Then, the first endpoint of $\partial(g)$ is covered by $g$ and by the previous guard. Similarly, the last endpoint of $\partial(g)$ is  covered by $g$ and  the next guard. 
\end{observation}

{\color{mycolor} A {\em segment} of $\partial$ is an edge or a portion of some edge of $\partial$. A {\em polygonal path} on $\partial$ is a contiguous portion  of $\partial$ where its endpoints can by anywhere on $\partial$. Let $s$ be a segment of $\partial$.} We define the  {\em covering region} of $s$, denoted $C(P,s)$, as the set of all points $p$ of $P$ where each point $p$ sees all the points of $s$. Notice that every point of $s$ and every point of $C(P,s)$ are visible to each other.\footnote{Also known as {\em complete visibility} where every point of $s$ and every point of $C$ are visible to each other~\cite{Avis1981}.} We define the covering region $C(P,\delta)$ of a polygonal path $\delta$ on $\partial$ analogously. If $p$ is a point in the covering region of $s$ (resp. $\delta$) then we say that $p$ {\em covers} $s$ (resp. $\delta$).
The following lemmas, though very simple, play important roles in our algorithms. Lemma~\ref{edge-covering-endpoints-obs} and Lemma~\ref{path-covering} can also be implied from ideas of \cite[Lemma~1 and Theorem~1]{Ghosh1991}; see also the companion paper \cite{Ghosh1996}.

\begin{lemma}
\label{edge-covering-endpoints-obs}
    The covering region of a polygonal path $\delta$ on $\partial$ is the intersection of the visibility polygons of all vertices of $\delta$. 
\end{lemma}

\begin{proof}
Since the covering region of $\delta$ is the intersection of the covering regions of all segments of $\delta$, it suffices to prove that for each segment $s\!=\!ab$ on $\delta$ we have $C(P,s)\!=\!V(P,a)\!\cap\! V(P,b)$.
To prove this we show that a guard $g$ covers $s$ if and only if $g$ sees $a$ and $b$.

For the forward direction, if $g$ covers $s$, it sees all points of $s$, including its endpoints $a$ and $b$.
    For the converse, if $g$ sees both $a$ and $b$, then the interior of the triangle $\bigtriangleup pab$ is in the interior of $P$, and thus $g$ sees all points on the segment~$s$.
\end{proof}

\begin{lemma}
\label{path-covering}
Let $\delta$ be a  polygonal path on $\partial$. Then, $C(P,\delta)$ is a simple polygon, and it can be computed in a polynomial time.
\end{lemma}

\begin{proof}
 To prove the first statement it suffices to show that for any two points $p$ and $q$ in $C{=}C(P,\delta)$, there is a path between $p$ and $q$ in $C$. Let $a$ and $b$ denote the endpoints of $\delta$.
    Let $C_p$ be the star-shaped polygon bounded by $\delta$ and the  rays $\overrightarrow{pa}$ and $\overrightarrow{pb}$. Define $C_q$ analogously.

    If $q\in C_p$ or $p\in C_q$, then any point on the line segment $pq$ is in $C$ because otherwise, the visibility of some point $x\in pq$ and some point $y\in \delta$ is blocked by some edge $e\in \delta$ in which case $y$ cannot be seen by $p$ or by $q$, a contradiction. Thus, $pq$ is a path in $C$.

    Assume that $q\notin C_p$ and $p\notin C_q$. In this case, a boundary ray of $C_p$ and a boundary ray of $C_q$ intersect at a point $c$, which is different from $a$ and $b$. We claim that $c$ is in $C$, because otherwise the visibility of $c$ and some point $y\in \delta$ is blocked by some edge $e\in\delta$, in which case $y$ cannot be seen by $p$ or $q$, a contradiction. Since $c\in C$, the previous argument implies that  the segments $pc$ and $cq$ are in $C$; these segments form a path between $p$ and $q$ in $C$.

    To prove the second statement, recall that $C(P,\delta)$ is the intersection of visibility polygons of all vertices of $\delta$ (by Lemma~\ref{edge-covering-endpoints-obs}). Since the visibility polygon from a point \cite{Lee1983} and the intersection of two polygons \cite{Berg2008} can be found in polynomial time, $C(P,\delta)$ can be computed in polynomial time by walking on $\delta$ and computing the intersection of the current covering region (which is a simple polygon by the first statement) with the visibility polygon of the next vertex on $\delta$.
\end{proof}

\begin{lemma}
\label{farthest-lemma}
Let $p$ be a fixed point on $\partial$. In polynomial time, we can find the farthest point $q$ from $p$ in counter-clockwise (or clockwise) direction along $\partial$ such that $\partial(p,q)$ can be covered by one guard. Such a guard can also be found in a polynomial time.
\end{lemma}

\begin{proof}
Due to symmetry, we explain the proof for the counter-clockwise direction. Let $e$ be the edge of $\partial$ that contains $p$, and let $p_0$ be the endpoint of $e$ after $p$. Denote the segment $pp_0$ by $e_0$. 
Find the longest polygonal path $\delta$ consisting of a sequence $e_0,e_1,\dots,e_t$ of consecutive edges of $\partial$ such that the covering region $C$ of $\delta$ is nonempty. In view of Lemma~\ref{path-covering}, this path can be found in polynomial time by iteratively considering the edges after $e_0$ and computing the covering region in each iteration. 

Let $e_{t+1}$ be the edge of $\partial$ after $e_t$. Our choice of $\delta$ implies that the covering region of the concatenation of $\delta$ and $e_{t+1}$ is empty. Moreover, $q$ lies in the interior or on an endpoint of $e_{t+1}$. Let $p_t$ and $p_{t+1}$ be the endpoints of $e_{t+1}$ where $p_{t+1}$ appears after $p_t$. Observe that every point of $C$ sees $p_t$, but no point of $C$ sees $p_{t+1}$. The point $q$ is then the farthest point from $p_t$ on $e_{t+1}$, that is visible from a point of $C$. In particular, as proven in the next paragraph, $q$ is the intersection point of $e_{t+1}$ and a ray from a vertex of $C$ through a reflex vertex of $P$. Thus $q$ can be found in polynomial time by checking all such rays.

Now we show that if $q$ is visible form a point $u\in C$ then it is visible from some vertex of $C$. The line through $e_{t+1}$ does not intersect $C$ because otherwise the intersection point  would see $p_{t+1}$.
Rotate $P$ such that the ray $\overrightarrow{p_{t+1}p_t}$ is upward. Assume that  $C$ is to the right side of the ray. 
Let $u'$ be the first point on the boundary of $C$ that is hit by the ray $\overrightarrow{qu}$. Let $e$ be the edge on the boundary of $C$ that contains $u'$, and let $v$ be the vertex of $e$ to the left of $\overrightarrow{qu}$. Notice that every point of $e$ sees $p_t$. Thus, the open quadrilateral $v,p_t,q,u'$ does not intersect $\partial$, and hence $v$ sees $q$ because $vq$ is a diagonal of the quadrilateral. If $C$ is to the left of $\overrightarrow{p_{t+1}p_t}$ then $q=p_t$ and a similar argument carries over.
\end{proof}

\section{A Greedy Algorithm}
\label{greedy-section}
In this section we present a greedy algorithm, for contiguous guarding of $\partial$, that finds a guard set $\Gamma$ of size at most $\OPT{+}1$.

If the intersection of the supporting halfplanes of the edges of $P$ is nonempty, then $P$ is star-shaped.\footnote{The {\em supporting halfplane} of an edge $e$ of $P$ is the halfplane its boundary goes through $e$ and it lies on the side of $e$ that is interior to $P$.} In this case, by placing a guard at the intersection we cover $\partial$, and thus $|\Gamma|=1$.

From now on assume that $P$ is not star-shaped.
Let $p_1$ be an arbitrary point on $\partial$, referred to as the {\em starting point}. Let $p_2\in \partial$ be the farthest point from $p_1$ in counter-clockwise direction such that $\partial(p_1,p_2)$ can be covered by one guard.  Let $p'_1\in \partial$ be the farthest point from $p_2$ in clockwise direction such that $\partial(p'_1,p_2)$ can be covered by one guard. See Figure~\ref{greedy-fig}. Observe that $p_1\in \partial(p'_1,p_2)$; it might be the case that $p'_1=p_1$.
Let $p_3,\dots,p_m$, with $m\ge 3$, be the points on $\partial$ such that 

\begin{enumerate}
    \item $p_{i}$ is the farthest point from $p_{i-1}$ in counter-clockwise direction such that $\partial(p_{i-1},p_i)$ can be covered by one guard, and
    \item $m$ is the smallest index for which $p'_1\in \partial(p_{m-1},p_m)$.
\end{enumerate}

We cover $\partial(p'_1,p_2)$ by a guard $g_1$, and for each $i\in\{3,4,\dots,m\}$ we cover $\partial(p_{i-1},p_i)$ by a guard $g_{i-1}$ as in Figure~\ref{greedy-fig}. In this case $|\Gamma|=m{-}1$. By Lemma~\ref{farthest-lemma} the points $p_2, p'_1, p_3,\dots, p_m$ and the guards $g_1,\dots,g_{m-1}$ can be computed in polynomial time.

\begin{figure}[htb]
	\centering
\setlength{\tabcolsep}{0in}
	$\begin{tabular}{cc}
	\multicolumn{1}{m{.48\columnwidth}}{\centering\vspace{0pt}\includegraphics[width=.448\columnwidth]{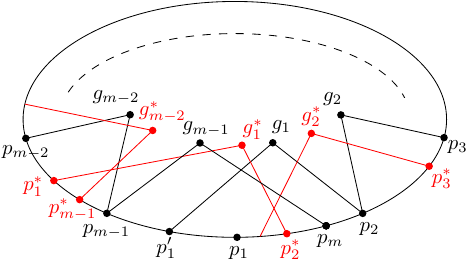}}
	&\multicolumn{1}{m{.52\columnwidth}}{\centering\vspace{0pt}\includegraphics[width=.52\columnwidth]{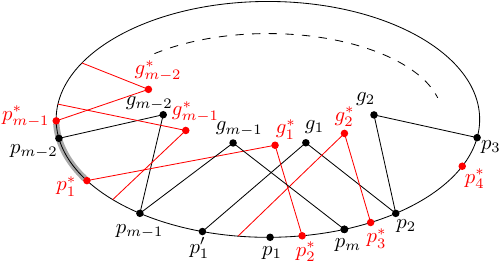}}
		\\
		(a)   &(b) 
	\end{tabular}$	
\caption{Illustration of the proof of Theorem~\ref{greedy-thr}. (a) $p_1$ is covered by one guard $g^*_1$ in an optimal solution. (b) $p_1$ is covered by two guards $g^*_1$ and $g^*_2$ in an optimal solution.}
\label{greedy-fig}
\end{figure}

\begin{theorem}
\label{greedy-thr}
    The Greedy Algorithm, starting from an arbitrary point $p_1$, runs in polynomial time and returns a guard set $\Gamma$ of size at most $\OPT{+}1$. Moreover, if $p_1$ is covered by two guards in some optimal solution, then  $|\Gamma|=\OPT$. 
\end{theorem}
\begin{proof}
By construction $\Gamma$ covers the entire $\partial$. 
Determining whether $P$ is star-shaped can be done in polynomial time by checking the intersection of the supporting halfplanes of edges. If $P$ is star-shaped, then $\Gamma$ has one guard in which case $|\Gamma|=\OPT$.

Assume that $P$ is not star-shaped, and hence $\OPT\ge 2$. Having $p_1$, we can find $p_2$ in polynomial time by Lemma~\ref{farthest-lemma}. Having $p_2$, we can find $p'_1$ and $g_1$ in polynomial time again by Lemma~\ref{farthest-lemma}. Similarly, for each $p_i$, with $i\ge 2$, we can find $p_{i+1}$ and $g_{i-1}$ in polynomial time. Thus, the algorithm runs in polynomial time.

Now we prove our claim about the size of $\Gamma$, which is $m-1${\color{mycolor}; the proof is a standard argument that shows the analogous statements for the circular arc cover problem.} First, assume that there is no optimal solution in which $p_1$ is covered by two guards. This case is illustrated in Figure~\ref{greedy-fig}(a). Consider some optimal solution and let $g^*_1$ be the only guard in this solution that covers $p_1$. Let $p^*_1, p^*_2$ be the first and last endpoints of $\partial(g^*_1)$, respectively. By our choice of $p_2$, the point $p^*_2$ appears on or before $p_2$. By our choice of $m$ and $p_{m-1}$, the point $p^*_1$ appears on or after $p_{m-2}$, because otherwise $p'_1$ would be in $\partial(p_{m-2},p_{m-1})$. For the next counter-clockwise guard in the optimal solution, say $g^*_2$, the last endpoint $p^*_3$ of $\partial(g^*_2)$ lies on or before $p_3$. Continuing this argument, there must be a guard $g^*_{m-2}$ in the optimal solution such that the last endpoint $p^*_{m{-}1}$ of $\partial(g^*_{m-2})$ appears on or before $p_{m-1}$. Therefore, the optimal solution has at least $m{-}2$ guards. Thus $|\Gamma|\le \OPT+1$.

Now assume that $p_1$ is covered by two guards, say $g^*_1$ and $g^*_2$, in some optimal solution, where $g^*_1$ appears before $g^*_2$; see Figure~\ref{greedy-fig}(b). We follow an argument similar to the previous case from $g^*_2$. By our choice of $p_2$, the last endpoint $p^*_3$ of $\partial(g^*_2)$ appears on or before $p_2$. Continuing this argument, the last endpoint of $p^*_{m-1}$ of $\partial(g^*_{m-2})$ appears on or before $p_{m-2}$. On the other hand, the guards $g^*_1$ and $g^*_2$ cannot cover $p_{m{-}2}$, because otherwise $g_{m-2}$ would cover $p_1$ and hence $p'_1$, contradicting our choice of $m$. Thus the optimal must have another guard $g^*_{m-1}$ to cover the gap between the coverages of $g^*_{m-2}$ and $g^*_1$. Therefore $|\Gamma|=\OPT=m{-}1$.
\end{proof}
The following is a direct corollary of Theorem~\ref{greedy-thr}.
\begin{corollary}
    \label{optimal-cor}
    The contiguous boundary guarding problem can be solved in polynomial time if we know a point of the boundary that is covered by two guards in some optimal solution.
\end{corollary}

\section{An Optimal Algorithm}
\label{optimal-section}
In this section we present a polynomial-time exact algorithm for contiguous guarding of the boundary $\partial$ of a simple polygon $P$. 
We assume that $P$ is not star-shaped and thus $\OPT\ge 2$.

Without loss of generality, we assume this maximality property: {\em any guard $g$  covers a maximal contiguous portion of $\partial$, i.e., $\partial(g)$ is maximal.}
This implies the following lemma.

\begin{lemma}
\label{on-ray-lemma}
   Let $g$ be a guard in a guarding of $\partial$ such that $\partial(g)$ is maximal. Then, each boundary ray of $w(g)$ goes through at least one reflex vertex of $P$. 
\end{lemma}  
\begin{proof}
First assume that $w(g)<2\pi$. The lemma holds because otherwise, we could increase the covering angle of $g$ and extend  $\partial(g)$, contradicting our maximality assumption. This is depicted in Figure~\ref{reflex-boundary-fig-a}(a).

Now assume that  $w(g)=2\pi$. Then the two boundary rays of $w(g)$ are identical. Moreover, they go through at least one reflex vertex because $P$ is not star-shaped and $g$ does not cover the entire $\partial$. This is depicted in Figure~\ref{reflex-boundary-fig-a}(b).
\end{proof}
\begin{figure}[htb]
$\begin{tabular}{cc}
	\multicolumn{1}{m{.45\columnwidth}}{\centering\vspace{0pt}\includegraphics[width=.28\columnwidth]{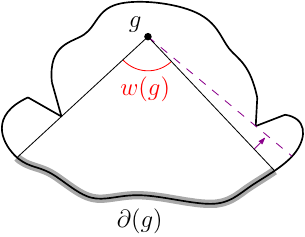}}
&\multicolumn{1}{m{.45\columnwidth}}{\centering\vspace{0pt}\includegraphics[width=.28\columnwidth]{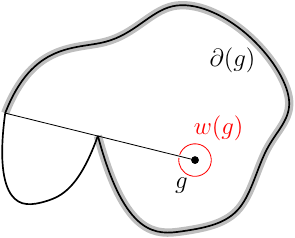}}
		\\
		(a)   &(b) 
\end{tabular}$
\caption{Illustration of the proof of Lemma~\ref{on-ray-lemma}.}
\label{reflex-boundary-fig-a}
\end{figure}

\subparagraph*{Algorithm {\normalfont (in a nutshell)}.}Our algorithm is fairly simple. It finds a polynomial-size set $S$ of starting points on $\partial$ such that at least one of the points is covered by two guards in some optimal solution. We execute the greedy algorithm for each point in $S$ and return the smallest guard set over all the points in $S$. 
From Corollary~\ref{optimal-cor}, it follows that the algorithm returns an optimal solution. The algorithm takes polynomial time because $S$ has a polynomial size, and the greedy algorithm runs in polynomial time.

\vspace{8pt}
It remains to compute $S$. This is the most crucial part of the algorithm. First, we find a polynomial-size point set $Q$ such that at least one guard in some optimal solution lies at a point of $Q$. For each point $q\in Q$, we compute a set $F(q)$ of all potential first endpoints of $\partial(q)$ for a guard at $q$. By Lemma~\ref{on-ray-lemma}, these points are the intersections of $\partial$ with the rays that start from $q$ and pass through reflex vertices of $P$. Thus, for every reflex vertex $r$ that is visible from $q$, we add the first intersection point of the ray $\overrightarrow{qr}$ with $\partial$ to $F(q)$. We then define the set $S$ as the union of $F(q)$ over all points $q\in Q$. The set $S$ has a polynomial size because $|Q|$ and the number of reflex vertices are bounded by polynomials.

Let $g^*$ be a guard in an optimal solution $\Gamma^*$ that lies at a point $q\in Q$. Then the first endpoint $f^*$ of $\partial(g^*)$ is in $F(q)$, and hence in $S$. By Observation~\ref{endpoint-obs}, $f^*$ is covered by $g^*$ and the previous guard in $\Gamma^*$. Therefore, $f^*$ is a point in $S$ with our desired property of being covered by two guards in an optimal solution.
\begin{figure}[!ht]
	\centering
\setlength{\tabcolsep}{0in}	\includegraphics[width=.43\columnwidth]{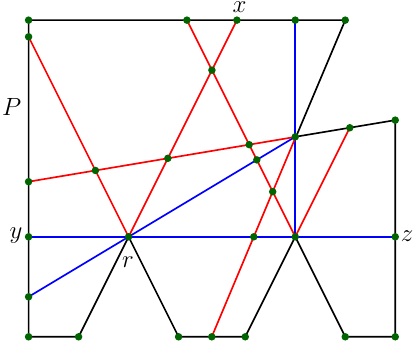}\\	\caption{Edge-extensions are  in red and vertex-extensions are in blue. The points in $Q$ are represented by small
    green disks.}
\label{Q-fig}
\end{figure}

Now it remains to compute $Q$.
For every reflex vertex $r$ of $P$, extend each edge incident to $r$ inside $P$ until hitting $\partial$ at some point $x$; see Figure~\ref{Q-fig}. We refer to the segment $rx$ as an {\em edge-extension}. Add a segment between every two reflex vertices of $P$ that are visible to each other and extend it from both endpoints inside $P$ until hitting $\partial$ at points $y$ and $z$. We refer to the segment $yz$ by {\em vertex-extension}. 
An {\em extension} is    either an edge-extension or a vertex-extension.
We define $Q$ as the set containing the following points

\begin{itemize}
    \item all vertices of $P$,
    \item all intersection points between extensions and $\partial$, 
    \item all intersection points between edge-extensions, 
    \item all intersection points between an edge-extension and a vertex-extension.
\end{itemize}

\noindent
The following structural lemma (proved in Section~\ref{proof-section}) shows that $Q$ has our desired property.

\begin{lemma}
  \label{structural-thr}
    Some optimal solution  has a guard at a point of $Q$.
\end{lemma}

If $P$ has $n$ vertices, then $Q$ has $O(n^3)$ points because there are $O(n)$ edge-extensions and $O(n^2)$ vertex-extensions. The set $S$ has size $O(n^4)$. 
The visibility polygon from a point can be computed in linear time \cite{Lee1983}, and the intersection of two polygons with $O(n)$ vertices can be found in $O(n\log n+k)$ time where $k$ is the size of the output  \cite{Balaban1995,Chazelle1992,Berg2008}.
For a given starting point $p_1$, the greedy algorithm finds $p_2,p_3,\cdots$ and $g_1,g_2,\cdots$ by computing the intersection of visibility polygons of $O(n)$ points ordered in counter-clockwise along $\partial$ (cf. Lemma~\ref{edge-covering-endpoints-obs}). The intersection of visibility polygons of consecutive points is a simple polygon (by Lemma~\ref{path-covering}) and it has $O(n)$ vertices because its reflex vertices coincide with reflex vertices of $P$. Thus the greedy algorithm takes $O(n^2\log n)$ time for a given starting point. Therefore, the total running time of the optimal algorithm, which invokes the greedy algorithm from each starting point in $S$, is $O(n^6\log n)$. This is a conservative upper bound on the running time and surely can be improved. 
For example, one might adopt ideas from the $O(n)$-time incremental algorithm for the kernel of a polygon \cite{Lee1979} to compute the covering region of a path on $\partial$ and consequently achieve a better running time for the greedy algorithm.
Our main goal here is a proof that the decision version of the contiguous boundary guarding problem is in the complexity class {\bf P}. The following theorem summarizes our result.

\begin{theorem}
The problem of contiguous boundary guarding of a polygon with the minimum number of guards can be solved in polynomial time. 
\end{theorem}


\subparagraph*{Remark.} One might wonder if Lemma~\ref{structural-thr} holds for every optimal solution, i.e., whether every optimal solution has a guard in $Q$. Also, one might wonder whether there always exists an optimal solution in which two guards cover the same vertex of the polygon. If this were true, then running the greedy algorithm from all vertices of the polygon would suffice. However, none of these statements are true as shown in Figure~\ref{not-in-Q-fig}. (For the second statement, one can verify that no guard can cover a red cross and a blue cross simultaneously. Therefore, any optimal solution must have a guard in the red triangle and a guard in the blue triangle. Then, the  contiguous coverages of such guards do not have any vertex in common.) 

\begin{figure}[!ht]
	\centering
\setlength{\tabcolsep}{0in}	\includegraphics[width=.45\columnwidth]{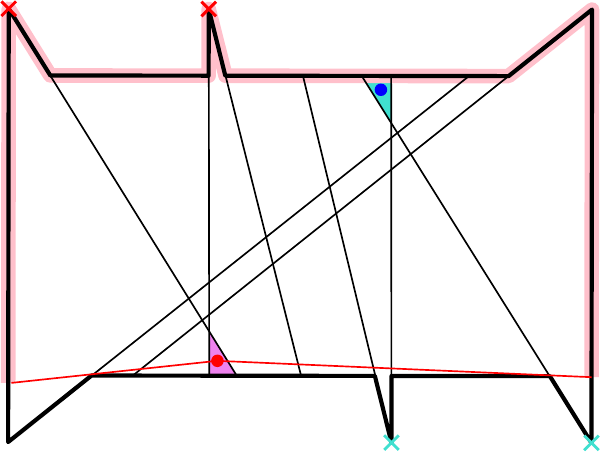}\\	\caption{An optimal solution, with two guards, none of which is in $Q$. The coverages of no two guards, in any optimal solution, share a vertex of the polygon.}
\label{not-in-Q-fig}
\end{figure}


\section{A Proof of Lemma~\ref{structural-thr}}
\label{proof-section}

We prove this lemma by contradiction. Assume that no optimal solution has a guard at a point of $Q$. We start with an overview of the proof.

\subparagraph*{Overview} We consider all optimal solutions where the coverage $\partial(g)$ of each guard $g$ is maximal, that is, we cannot extend $\partial(g)$ by increasing $w(g)$ or by moving $g$. Among such optimal solutions we choose $\Gamma^*$ as one in which the covering angle of each guard is maximal. Let $g_1,g_2,\dots,g_{m}$ be the guards in $\Gamma^*$ in this order along $\partial$. For each $g_i$ we show that $\pi\le w(g_i)<2\pi$ (Lemma~\ref{2pi-lemma} and Lemma~\ref{angle-lemma}) and $g_i$ does not lie on any vertex-extension (Lemma~\ref{no-vertex-extension}). Moreover we show that each $g_i$ lies on the extension of exactly one edge (Lemma~\ref{edge-lemma}), denoted $e(g_i)$, and that moving $g_i$ on the extension towards or away from $e(g_i)$ could make either endpoint of $\partial(g_i)$ shorter. We employ this last property and move the guards such that for each pair of  consecutive guards $g_i$ and $g_{i+1}$ the intersection of $\partial(g_i)$ and $\partial(g_{i+1})$ is exactly one point, except possibly for the pair $(g_m,g_1)$. The resulting guard set, denoted $\Gamma$, is an optimal solution and has the same coverage as $\Gamma^*$. In $\Gamma$ we could either move $g_1$ to lie on a point of $Q$ (when $\OPT=2$) or conclude that $\Gamma$, and hence $\Gamma^*$, do not cover the entire $\partial$ (when $\OPT>2$). Both cases lead to a contradiction.

\vspace{8pt}
Now we present the details of the proof. Let $\cal{G}$ be the set of all optimal guard sets for $P$ such that each guard has a maximal coverage. For a guard set $\Gamma\in\cal{G}$, we define the {\em angle  sequence}, denoted $\angseq{\Gamma}$, as the sequence of covering angles of the guards in $\Gamma$ sorted in decreasing order (the angles are measured within $[0,2\pi]$). Then, we define an order on the elements of $\cal{G}$ as follows: for $\Gamma_1,\Gamma_2\in\cal{G}$, we say that $\Gamma_1 \succ \Gamma_2$ if and only if $\angseq{\Gamma_1}\succ\angseq{\Gamma_2}$ in the lexicographic order. Let $\Gamma^*$ be a maximal element of $\cal{G}$ with respect to $\succ$.

We refer to the boundary rays of a wedge $w(g)$ simply by the {\em rays} of $w(g)$. We say that a ray is {\em bounded} by a reflex vertex if the ray passes through a reflex vertex. By Lemma~\ref{on-ray-lemma}, for each $g\in \Gamma^*$, each ray of $w(g)$ is bounded by some reflex vertex.

\begin{lemma}
\label{2pi-lemma}
    For any guard $g\in\Gamma^*$ in the interior of $P$ it holds that $w(g)<2\pi$.
\end{lemma}
\begin{proof}
For the sake of contradiction assume that $w(g)=2\pi$ for some guard $g$. Then the boundary rays of $w(g)$ are identical and go through at least one reflex vertex (by Lemma~\ref{on-ray-lemma}). If the rays go through two reflex vertices then by moving $g$ along the ray either $g$ reaches an edge-extension or a reflex vertex (without decreasing its coverage or its angle); this case is depicted in Figure~\ref{pi-fig}(a). Thus we get a new optimal solution that contradicts the fact that $g$ is not in $Q$. Assume that the rays go through exactly one reflex vertex $r$. If $g$ is not on any edge-extension we can rotate $g$ around $r$ and increase its coverage $\partial(g)$, contradicting our choice of $\Gamma^*$. If $g$ is on two edge-extensions then it is in $Q$, a contradiction. Assume that $g$ is on an extension of exactly one edge $e$, as in Figure~\ref{pi-fig}(b). Then by moving $g$ along the extension, towards or away from $e$, we increase the coverage of $g$, a contradiction. 
\end{proof}

\begin{figure}[htb]
	\centering
\setlength{\tabcolsep}{0in}
	$\begin{tabular}{cc}
	\multicolumn{1}{m{.5\columnwidth}}{\centering\vspace{0pt}\includegraphics[width=.32\columnwidth]{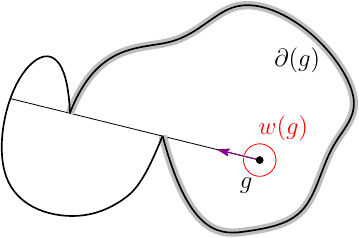}}
&\multicolumn{1}{m{.5\columnwidth}}{\centering\vspace{0pt}\includegraphics[width=.28\columnwidth]{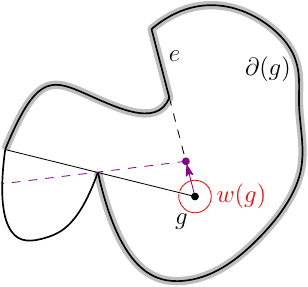}}
		\\
		(a)   &(b) 
	\end{tabular}$	
\caption{Illustration of the proof of Lemma~\ref{2pi-lemma}.}
\label{pi-fig}
\end{figure}

\begin{lemma}
\label{edge-lemma}
    Any guard $g\in\Gamma^*$ that is in the interior of $P$  lies on exactly one edge-extension. Moreover, the corresponding edge belongs to $\partial(g)$.
\end{lemma}
\begin{proof}
    First, observe that $g$ cannot be on two edge-extensions because otherwise $g$ is in $Q$, a contradiction.
    If $g$ is not on any edge-extension in $\partial(g)$ as in Figure~\ref{reflex-boundary-fig}, then by moving $g$ along the bisector of $w(g)$ one can increase $w(g)$ without decreasing its coverage $\partial(g)$, a contradiction to our choice of $\Gamma^*$---this movement is possible as $g$ is in the interior of $P$ and $w(g)<2\pi$ by Lemma~\ref{2pi-lemma}. Thus $g$ lies on the extension of exactly one edge, say $e$. If $e\notin \partial(g)$ then we can get a similar contradiction by slightly moving $g$ towards the bisector of $w(g)$.
\end{proof}

\begin{figure}[htb]
	\centering
\includegraphics[width=.3\columnwidth]{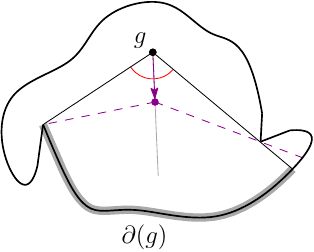}
\caption{Illustration of the proof of Lemma~\ref{edge-lemma}.}
\label{reflex-boundary-fig}
\end{figure}

\begin{lemma}
    \label{no-vertex-extension}
    No guard $g\in \Gamma^*$ lies on a vertex-extension.
\end{lemma}
\begin{proof}
If $g\in \partial$, then $g$ cannot be on a vertex-extension because otherwise, it would be in $Q$, a contradiction.
If $g\notin \partial$, then $g$ is in the interior of $P$. In this case, $g$ is on an edge-extension, by Lemma~\ref{edge-lemma}, and thus it cannot be on a vertex-extension as otherwise it would be in $Q$. 
\end{proof}

A guard $g\in\Gamma^*$ cannot lie on both $\partial$ and an edge-extension because otherwise it would be in $Q$, a contradiction. 
Moreover, a guard cannot be on a vertex of the polygon. This and Lemma~\ref{edge-lemma} imply that every guard $g$ lies either on exactly one edge of $\partial$ or on exactly one edge-extension but not both. For both cases, we denote the corresponding edge by $e(g)$. (In the case where $g$ lies on an edge-extension, $e(g)$ is the edge of $P$ that defines the extension.) 


Let $r$ be a reflex vertex that bounds a ray of $w(g)$; such a reflex vertex exists by Lemma~\ref{on-ray-lemma}. We say that  $r$ bounds the ray from {\em outside} if both edges incident to $r$ lie outside $w(g)$; otherwise (at least one edge lies inside $w(g)$), we say that $r$ bounds the ray from {\em inside}.

\begin{figure}[htb]
	\centering
\setlength{\tabcolsep}{0in}
	$\begin{tabular}{cc}
	\multicolumn{1}{m{.5\columnwidth}}{\centering\vspace{0pt}\includegraphics[width=.4\columnwidth]{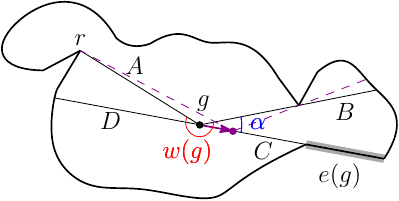}}
&\multicolumn{1}{m{.5\columnwidth}}{\centering\vspace{0pt}\includegraphics[width=.32\columnwidth]{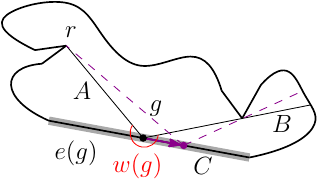}}
		\\
		(a)   &(b) 
	\end{tabular}$	
\caption{Illustration of the proof of Lemma~\ref{reflex-inside-lemma}.}
\label{bounded-inside-fig}
\end{figure}

\begin{lemma}
\label{reflex-inside-lemma}
For any guard $g\in\Gamma^*$,  each ray of $w(g)$ is bounded by exactly one reflex vertex of $P$. Such a vertex bounds the ray from outside.
\end{lemma} 
\begin{proof}
This first statement is a direct implication of Lemma~\ref{on-ray-lemma} and Lemma~\ref{no-vertex-extension}.

We prove the second statement by   contradiction. Assume that a reflex vertex $r$ bounds a ray $A$ of $w(g)$ from inside. Then $r$ is an endpoint of $\partial(g)$. Let $B$ denote the other ray of $w(g)$, which is bounded from inside or outside as in Figure~\ref{bounded-inside-fig}.
We consider two cases:

\begin{itemize}
    \item $g\notin \partial$. Then by Lemma~\ref{edge-lemma}, $g$ lies on the extension of $e(g)$, and moreover $e(g)$ belongs to $\partial (g)$. 
Let $C$ and $D$ be the rays from $g$ towards and away from $e(g)$, respectively, as in Figure~\ref{bounded-inside-fig}(a).  {\color{mycolor} Let  $\alpha$ be the angle between $B$ and $C$ that lies inside $w(g)$. If $\alpha$} is at most $\pi$, then by moving $g$ on $C$  either $g$ reaches a vertex-extension, or another edge-extension, or an endpoint of $e(g)$. This movement does not decrease $\partial(g)$ because the endpoint $r$ of $\partial(g)$ remains visible from $g$ as there is no reflex vertex in the interior of $A$, and the other endpoint of $\partial(g)$, which is on $B$, remains visible since the movement angle  $\alpha$ is at most $\pi$. In all cases we get a new optimal solution that contradicts the fact that $g$ is not in $Q$. Assume that $\alpha> \pi$ {\color{mycolor}(this could happen for example if in Figure~\ref{bounded-inside-fig}(a) the edge $e(g)$ and consequently the rays $C$ and $D$ were reflected around $g$).} By moving $g$ on $D$, either $g$ lies on a vertex-extension, or on another edge-extension, or it reaches $\partial$. In all cases, we get an optimal solution for which $g$ is in $Q$, which is a contradiction.

 \item $g\in \partial$. If $w(g)\leq \pi$ then by moving $g$ along the bisector of $w(g)$ we increase the covering angle of $g$ without decreasing $\partial(g)$, a contradiction to our choice of $\Gamma^*$. Assume that $w(g)> \pi$. Let $C$ be the ray from $g$ along  $e(g)$ such that the angle between $B$ and $C$ inside $w(g)$ is at most $\pi$ as in Figure~\ref{bounded-inside-fig}(b). By moving $g$ on $e(g)$ along $C$, either $g$ lies on a vertex-extension, or on another edge-extension, or it reaches an endpoint of $e(g)$. As in the previous case, this movement does not decrease $\partial(g)$. In all cases, we get an optimal solution for which $g$ is in $Q$, a contradiction. 
\qedhere
\end{itemize}
\end{proof}

\begin{lemma}
\label{angle-lemma}
For every guard $g\in\Gamma^*$, it holds that $w(g)\ge \pi$. 
\end{lemma}
\begin{proof}
First, assume that $g\notin e(g)$, and hence, it lies in the interior of $P$.
By Lemma~\ref{reflex-inside-lemma}, each ray of $w(g)$ is bounded by exactly one reflex vertex from outside.
If $w(g)< \pi$ then by  moving $g$ inside $w(g)$, along the extension of $e(g)$, we increase $w(g)$ without decreasing $\partial(g)$ as in Figure~\ref{edge-fig}(a). This contradicts our choice of $\Gamma^*$.

Now assume that $g\in e(g)$. If $w(g)< \pi$ then by moving $g$ along the bisector of $w(g)$ we increase its covering angle without decreasing $\partial(g)$, a contradiction to our choice of $\Gamma^*$.  
\end{proof}

\begin{figure}[htb]
	\centering
\setlength{\tabcolsep}{0in}
	$\begin{tabular}{cc}
	\multicolumn{1}{m{.5\columnwidth}}{\centering\vspace{0pt}\includegraphics[width=.32\columnwidth]{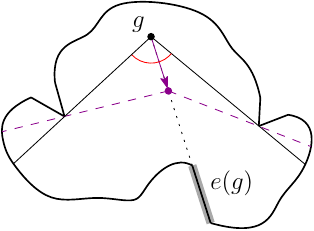}}
	&\multicolumn{1}{m{.5\columnwidth}}{\centering\vspace{0pt}\includegraphics[width=.37\columnwidth]{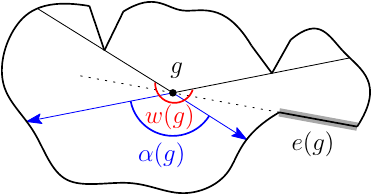}}
    \\
		(a)   &(b) 
	\end{tabular}$	
\caption{Illustration of (a) the proof of Lemma~\ref{angle-lemma}, and (b) the statement of Lemma~\ref{edge-out-lemma}.}
\label{edge-fig}
\end{figure}

Let $\alpha(g)$ be the wedge with apex $g$ that is obtained by extending boundary rays of $w(g)$ as in Figure~\ref{edge-fig}(b). Since $w(g)\ge \pi$, the wedge  $\alpha(g)$ lies inside $w(g)$. Moreover, $\alpha(g)\le \pi$. 

\begin{lemma}
\label{edge-out-lemma}
For any guard $g\in\Gamma^*$, the edge $e(g)$ is inside $w(g)$ but outside $\alpha(g)$.
\end{lemma}
\begin{proof}
If $g\in e(g)$, then the statement holds because both rays of $w(g)$ are on the same side of $e(g)$. Assume that $g\notin e(g)$.  
It follows from Lemma~\ref{edge-lemma} that $e(g)\in w(g)$. To verify that $e(g)\notin \alpha(g)$ observe that otherwise, we could move $g$ along the extension of $e(g)$ inside $\alpha(g)$ to increase $w(g)$ without decreasing $\partial(g)$, a contradiction to our choice of $\Gamma^*$.  
\end{proof}

Let $g_1,g_2,\dots,g_m$ be the guards in $\Gamma^*$ in this order along $\partial$ (recall that this is the order of the unique points they cover). We modify $\Gamma^*$ and obtain an optimal solution $\Gamma$ such that:

\begin{itemize}
    \item for every $g_i\in \Gamma$ it holds that $w(g_i)\ge \pi$ and $e(g_i)\in w(g_i){\setminus}\alpha(g_i)$, and 
    \item for each pair of consecutive guards $g_i$ and $g_{i+1}$ the intersection of $\partial(g_i)$ and $\partial(g_{i+1})$ is exactly one point, except possibly for the pair $(p_m,p_1)$.
\end{itemize}

\begin{figure}[htb]
	\centering
\setlength{\tabcolsep}{0in}
	$\begin{tabular}{cc}
	\multicolumn{1}{m{.5\columnwidth}}{\centering\vspace{0pt}\includegraphics[width=.35\columnwidth]{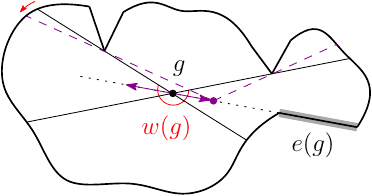}}
	&\multicolumn{1}{m{.5\columnwidth}}{\centering\vspace{0pt}\includegraphics[width=.35\columnwidth]{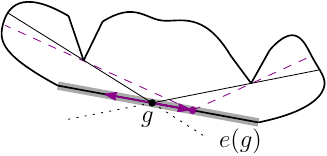}}
		\\
		(a)   &(b) 
	\end{tabular}$	
\caption{Illustration of the movement of $g$ along the extension of $e(g)$ or on $e(g)$.}
\label{movement-fig}
\end{figure}

By Lemma~\ref{edge-out-lemma} for every $g_i\in \Gamma^*$ we have $e(g_i)\in w(g_i){\setminus}\alpha(g_i)$. Assume that $g\notin e(g_i)$. Consider moving $g_i$ along the extension of $e(g_i)$; see Figure~\ref{movement-fig}(a). If we move $g_i$ towards $e(g_i)$ then one endpoint of $\partial(g_i)$ gets shorter, and if we move $g_i$ away from $e(g_i)$ the other endpoint of $\partial(g_i)$ gets shorter. This means that we can move $g_i$ along the extension to make either one of $\partial(g_i)\cap\partial(g_{i+1})$ and $\partial(g_i)\cap\partial(g_{i-1})$ a point. Along this movement, $g_i$ cannot reach $e(g_i)$ or the boundary of $P$ because otherwise, $g_i$ would be in $Q$, a contradiction. Also, $g_i$ cannot lie on a new edge-extension or vertex-extension as it would be in $Q$ again; in particular, this means that the boundary rays of $w(g_i)$ cannot hit a new reflex vertex.

When $g\in e(g_i)$ we can make either one of $\partial(g_i)\cap\partial(g_{i+1})$ and $\partial(g_i)\cap\partial(g_{i-1})$ a point by moving $g_i$ on $e(g_i)$ towards an endpoint as in Figure~\ref{movement-fig}(b). Again, along this movement, $g_i$ cannot reach an endpoint or lie on a vertex-extension or on an edge-extension.

After the movement, the angle at $w(g_i)$ remains at least $\pi$ because both bounding reflex vertices remain on the same side of the extension of $e(g_i)$. This also implies that $e(g_i)$ remains inside $w(g_i)$ and outside $\alpha(g_i)$.

We obtain $\Gamma$ as follows. We fix $g_1$ and move $g_2$ to make $\partial(g_1)\cap\partial(g_2)$ a point. Then fix $g_2$ and move $g_3$ to make $\partial(g_2)\cap \partial(g_3)$ a point. Then we perform this for all guards $g_3, g_4,\dots,g_{m-1}$ in this order. After this process all intersections $\partial(g_i) \cap \partial(g_{i+1})$ are points, except possibly for $\partial(g_m)\cap\partial(g_1)$. The resulting guard set is $\Gamma$.

For the rest of the proof, we consider two cases where $\OPT=2$ and $\OPT\ge 3$ separately and get a contradiction in each case.

\begin{figure}[!ht]
	\centering
\setlength{\tabcolsep}{0in}	\includegraphics[width=.58\columnwidth]{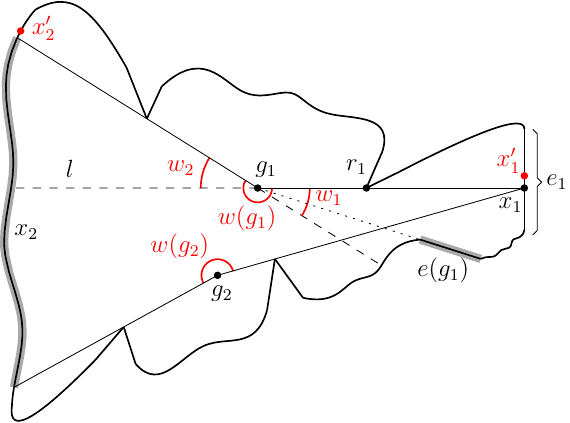}\\	\caption{Illustration of the proof for the case $\OPT=2$.}
\label{opt-2-revised-fig}
\end{figure}

\subparagraph*{Case 1 ($\OPT=2$)}
In this case $\partial(g_1)$ and $\partial(g_2)$ intersect at their both endpoints. One of the intersections is a point $x_1$ on some edge $e_1$, and the other is a portion $x_2$ of $\partial$; see Figure~\ref{opt-2-revised-fig}. Let $l$ be the line through $g_1x_1$, and let $r_1$ be the reflex vertex on $g_1x_1$. After a suitable rotation, assume that $l$ is horizontal and $r_1$ bounds $w(g_1)$ from above. Consider a point $x'_1$ on $e\setminus\partial(g_1)$ close to $x_1$---$x'_1$ is outside $w(g_1)$. Consider an analogous point $x'_2$ on $\partial\setminus \partial(g_1)$ close to $x_2$---$x'_2$ is outside $w(g_1)$. Both $x'_1$ and $x'_2$ are covered by $g_2$. Thus $g_2$ lies in $\alpha(g)$ as otherwise it would not cover $x'_1$ or $x'_2$. One ray of $w(g_2)$ intersects $\partial$ at $x_1$, and since $w(g_2)\ge \pi$ (by Lemma~\ref{angle-lemma}) its other ray intersects $\partial$ below $l$.
Let $w_1$ and $w_2$ be the two opposite wedges of $w(g_1)\setminus\alpha(g_1)$ where $w_1$ contains $x_1$. We have the following three cases depending on the position of $g_1$ and $e(g_1)$.

\begin{itemize}
    \item $g_1\notin e(g_1)$ and $e(g_1)\in w_1$. This case is depicted in Figure~\ref{opt-2-revised-fig}. By moving $g_1$ in $w_1$ towards $e(g_1)$, either $g_1$ lies on vertex-extension, or on a new edge-extension, or it reaches $e(g_1)$. In all cases we get a new optimal solution for which $g_1\in Q$, a contradiction. The movement cannot make the intersection $x_2$ empty because the intersection point of $l$ and $\partial$ remains in the coverage of $g_1$ (unless $g_1$ lies on a vertex-extension or a new edge extension during the movement). Therefore $\partial$ remains covered by $g_1$ and $g_2$ after the above movement, and thus the contradiction is valid.
    \item $g_1\notin e(g_1)$ and $e(g_1)\in w_2$. Again, by moving $g_1$ in $w_1$ away from $e(g_1)$, either $g_1$ lies on vertex-extension, or on a new edge-extension, or it reaches $\partial$. In all cases we get a new optimal solution for which $g_1\in Q$, a contradiction.
    \item $g_1\in e(g_1)$. By moving $g_1$ in $w_1$ on $e(g_1)$, either $g_1$ lies on vertex-extension or on a new edge-extension, or it reaches an endpoint of $e(g_1)$. Again, we get an optimal solution with $g_1\in Q$, a contradiction.
\end{itemize}
\begin{figure}[!ht]
	\centering
\setlength{\tabcolsep}{0in}	\includegraphics[width=.48\columnwidth]{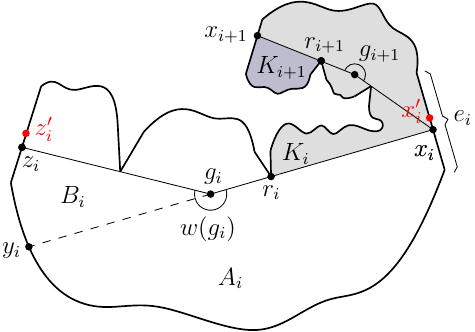}\\	\caption{Illustration of the proof for the case $\OPT\ge 3$.}
\label{opt-3-fig}
\end{figure}

\subparagraph*{Case 2 ($\OPT\ge 3$)} For each $i=1,\dots,m{-}1$ let $x_i$ be the point shared by $\partial(g_i)$ and $\partial(g_{i+1})$, $e_i$ be the edge of $P$ that contains $x_i$, and $r_i$ be the reflex vertex on $g_ix_i$ that bounds $w(g_i)$. See Figure~\ref{opt-3-fig}. Let $K_i$ be the portion of the polygon outside the coverage of $g_i$ that is cut off by $r_ix_i$. Let $y_i$ be the first point of $\partial$ hit by the ray $\overrightarrow{x_ig_i}$. It is implied from Lemma~\ref{angle-lemma} that $y_i\in \partial(g_i)$. The segment $x_iy_i$ partitions the polygon into three parts, one of which is $K_i$. Denote the part that lies completely in $w(g_i)$ by $A_i$ and the third part by $B_i$. 

Let $x'_i$ be a point on $e_i\cap K_i$ very close to $x_i$. The point $x'_i$ is not covered by $g_i$ but by $g_{i+1}$. The guard $g_{i+1}$ cannot be in $B_i$ because otherwise it cannot cover $x'_i$. It cannot be in $A_i$ either because otherwise $g_i$ and $g_{i+1}$ cover the entire polygon (recall that $w(g_{i+1})\ge \pi$), contradicting that optimal has size at least 3. Therefore $g_{i+1}$ is in $K_i$, and in particular $g_m$ is in $K_{m-1}$. The coverage $\partial(g_{i+1})$ of $g_{i+1}$ also lies in $K_i$ because its first endpoint is $x_i$ and its last endpoint cannot go beyond $r_i$. This implies that $K_{i+1}$ is a subset of $K_i$. Therefore $K_{m-1}$ is in $K_1$, and hence $g_m$ lies in $K_1$.

Let $z_1$ be the endpoint of $\partial(g_1)$ different from $x_1$. Notice that $z_1\in B_1$.
Let $z'_1$ be a point in $B_1$ and on $\partial\setminus \partial(g_1)$ very close to $z_1$. This point must be covered by $g_m$. However, this is impossible as $g_m$ is in $K_1$ and its coverage is blocked by $r_1$, a contradiction that the entire $\partial$ is covered by $\Gamma$.

\section{Conclusions}
We presented a polynomial-time exact algorithm as well as tight combinatorial bounds for the contiguous boundary guarding problem. One natural question is to improve the running time, say, by reducing the size of the set $Q$, or the size of the starting point set $S$, or the time of the greedy algorithm by adapting ideas from computing the kernel of a polygon \cite{Lee1979}.

Since polygon-guarding problems are typically NP-hard, it would be interesting to identify other instances that can be solved in polynomial time.




\subparagraph*{Acknowledgements}We thank the reviewers of SoCG 2025 who meticulously verified
our proofs and provided valuable feedback that clarifies some aspects of the proofs.
\bibliography{Boundary-Guarding.bib}
\end{document}